\documentclass[12pt]{article}

\usepackage{amsmath, amssymb, amsfonts, amsbsy, amsthm, latexsym, color}
\usepackage{graphicx}
\usepackage{enumerate}
\usepackage{authblk}
\everymath{\displaystyle}
\newtheorem{theorem}{Theorem}[section]

\newtheorem{lemma}[theorem]{Lemma}
\newtheorem{corollary}[theorem]{Corollary}

\theoremstyle{definition}
\newtheorem{definition}[theorem]{Definition}
\newtheorem{example}[theorem]{Example}

\newtheorem{remark}[theorem]{Remark}

\newcommand{\Ac}{\mathcal{A}}
\newcommand{\Wc}{\mathcal{W}}

\newcommand{\Uc}{\mathcal{U}}
\newcommand{\R}{\mathbb{R}}
\newcommand{\C}{\mathbb{C}}
\newcommand{\E}{\mathbb{E}}
\newcommand{\Ec}{\mathcal{E}}
\newcommand{\Fc}{\mathcal{F}}
\newcommand{\Hc}{\mathcal{H}}
\newcommand{\Nc}{\mathcal{N}}
\newcommand{\Mc}{\mathcal{M}}
\newcommand{\Pc}{\mathcal{P}}
\newcommand{\e}{\varepsilon}

\newcommand{\vol}{\mathrm{vol}}
\newcommand{\Sbd}{\mathbb{S}^{d-1}}
\newcommand{\vphi}{{\varphi}}

\newcommand{\vf}{{f}}
\newcommand{\vx}{{ x}}

\DeclareMathOperator*{\argmin}{arg\,min}

\begin{document}

\title{A Unified Recovery of Structured Signals Using Atomic Norm}

\author{Xuemei Chen}
\affil{\small Department of Mathematics and Statistics, University of North Carolina Wilmington, NC 28405}

%
%
%


\date{\vspace{-5ex}}


%

\maketitle
\begin{abstract}
In many applications we seek to recover signals from linear measurements far fewer than the ambient dimension, given the signals have exploitable structures such as sparse vectors or low rank matrices. In this paper we work in a general setting where signals are approximately sparse in a so called atomic set. We provide general recovery results stating that a convex programming can stably and robustly recover signals if the null space of the sensing map satisfies certain properties. Moreover, we argue that such null space property can be satisfied with high probability if each measurement is subgaussian even when the number of measurements are very few. Some new results for recovering signals sparse in a frame, and recovering low rank matrices are also derived as a result.
\end{abstract}

%
%


\section{Introduction}
Given a  compact set $E\subset\R^d$, let $\text{conv}(E)$ be the convex hull of $E$. We define the \emph{gauge} of E to be
\begin{equation}\label{equ:m}
g_E(v):=\inf\{\lambda>0: \lambda^{-1}v\in \text{conv}(E)\}.
\end{equation}
If $E=-E$, then \eqref{equ:m} can be rewritten as 
\begin{equation}\label{equ:m2}
\inf\{\sum_{i\in I}c_i: v = \sum_{i\in I}c_ix_i, c_i\geq0, x_i\in E\}
\end{equation}
with a straightforward argument. A gauge becomes a norm if it is finite everywhere, symmetric, and positive except at the origin~\cite[page 130]{R97}.

Throughout this paper, we will assume  the \emph{atomic set} $\Wc=\{w_i\}_{i\in I}$ is a compact spanning subset of $\R^d$. 
Define $\Wc_{\text{sym}}:=\Wc\cup(-\Wc)$. The gauge function associated with $\Wc_{\text{sym}}$ is now a norm,  called the \emph{atomic norm}~\cite{C12}. Given \eqref{equ:m} and \eqref{equ:m2},  the atomic norm associated with $\Wc$ is defined as
\begin{equation}\label{equ:norm}
\|v\|_\Wc:=\inf\{\lambda>0: \lambda^{-1}v\in \text{conv}(\Wc_{\text{sym}})\}=\inf\{\sum_{i\in I}|c_i|: v = \sum_{i\in I}c_iw_i\}.
\end{equation}

Let $C_\Wc$ be the smallest constant such that
$$\|v\|_\Wc\leq C_\Wc\|v\|_2, \text{ for every v}\in\R^d.$$
We let $B_\Wc:=\{v\in\R^d:\|v\|_\Wc\leq1\}=\text{conv}(\Wc_{\text{sym}})$ be the unit ball with respect to the atomic norm.


Let $\Ac$ be a linear mapping from $\R^d$ to $\R^m$ where $m$ tends to be a lot smaller than $d$.
We would like to recover a structured signal $z_0$ from few linear measurements $y=\Ac z_0$ via the convex minimization problem
\begin{equation}\label{equ:K0}\tag{$\text{P}_\Wc$}
\hat{z}=\arg\min_{z\in \R^d} \|z\|_{\Wc} \quad\text{ subject to } \quad \Ac z=y.\\
\end{equation}

If the measurement is perturbed as $y=\Ac z_0+e$ where $\|e\|_2\leq\epsilon$, then we consider the following version
\begin{equation}\label{equ:K}\tag{$\text{P}_{\Wc,\epsilon}$}
\hat{z}=\arg\min_{z\in \R^d} \|z\|_{\Wc} \quad\text{ subject to } \quad \|\Ac z-y\|_2\leq \epsilon.\\
\end{equation}

By ``structured signal'', we mean that it is a linear combination of very few atoms in $\Wc$.
We say a vector $z$ is \emph{$\Wc$-sparse of order $s$ (or simply $\Wc$-$s$-sparse)} if $z$ can be written as a linear combination of at most $s$ atoms from $\Wc$. Such collection of vectors will be denoted $\Sigma_{\Wc,s}$, i.e., $$\Sigma_{\Wc,s}:=\{z\in\R^d: z = \sum_{i\in J}c_iw_i\text{ for some }|J|\leq s, c_i\in\R\}.$$ 

Recovering $z_0$ through \eqref{equ:K} is possible even when the number of linear measurement $m$ is far less than the ambient dimension $d$, if  $\Ac$ is well designed. 

We  need a way to define the ``tail'' when a vector is not exactly $\Wc$-sparse. This will be an important notion in this paper. Define
\begin{equation}\label{equ:tWW}
\sigma_{\Wc,s}(z):=\inf\{\|z-v\|_{\Wc}: v\in\Sigma_{\Wc,s}\}.
\end{equation}
This tail can also be viewed as the $\Wc$-distance from $z$ to the set $\Sigma_{\Wc,s}$.

In the case that the minimum can be obtained in \eqref{equ:tWW}, we write
\begin{equation}\label{equ:zs}
z_s:=\argmin\limits_{v\in\Sigma_{\Wc,s}}\|z-v\|_{\Wc},
\end{equation}
and hence $\sigma_{\Wc,s}(z)=\|z-z_s\|_\Wc$. Note that there may be multiple minimizers of \eqref{equ:zs}.

\subsection{Notations}
Given the linear map $\Ac$, $\Nc(\Ac)$ is the null space of $\Ac$, and $\nu_\Ac$ is the smallest nonzero singular value of the matrix of map $\Ac$ under an orthonormal basis. For any matrix $X$, $\{\sigma_i(X)\}_{i=1}^K$ are the singular values of $X$, and $\|X\|_*=\sum_{i=1}^K\sigma_i(X)$ is the nuclear norm of $X$. $\|X\|_F$ is the Frobenius norm of $X$.
For any index set $T$, $|T|$ is the cardinality of $T$, and $T^c$ is the compliment set of $T$.
For any $p\geq1$, $\|\cdot\|_p$ is the $\ell_p$ norm.
$\Sbd$ is the unit $\ell_2$ norm ball of $\R^d$.

The notation $\R^d$  means a finite dimensional real Hilbert space in general, but also denotes the real Euclidean space of dimension $d$ in certain examples which will be clear from the context.

\subsection{Examples}
We will now list several examples of atomic sets that are special cases of our setup. See more examples in \cite{C12}.
\begin{example}[Compressed Sensing]\label{exa:cs}
As the simplest example, when the atom set $\Ec$ is the canonical orthonormal basis of the Euclidean space, we have $\|v\|_\Ec=\|v\|_1$. In this context, an $\Ec$-$s$-sparse vector is a vector whose coordinates are nonzero  at most $s$ locations. This is the classical compressed sensing~\cite{CRT06, D06, CDD09}. It is easy to show that the tail $\sigma_{\Ec,s}(z)=\|z_{T^c}\|_1$ where $T$ is the index set of the largest $s$ coordinates in magnitude of $z$.
\end{example}

\begin{example}[Compressed Sensing with Frames]\label{exa:frame}
Let $\Fc=\{f_i\}_{i=1}^N$ be a frame of the Euclidean space $\R^d$ and $\Fc$ will be our atomic set. A frame for a finite dimensional Hilbert space is simply a spanning set of it, hence $N\geq d$. Let $F=[f_1,f_2,\cdots,f_N]$ be the $d\times N$ matrix whose columns are the atoms.

First, for any $z\in\R^d$, $\|z\|_\Fc=\inf\{\|x\|_1: Fx=z\}=\min\{\|x\|_1: Fx=z\}$. This is because $\inf\{\|x\|_1: Fx=z\} = \inf\{\|x\|_1: Fx=z, \|x\|_1\leq\|F^\dagger z\|_1\} $ since $FF^\dagger z=z$, and a continuous function obtain its minimum on a compact set.

So for any $z, v\in\R^d$, we have
\begin{align}\notag
\|z-v\|_\Fc&=\min\{\|c\|_1: Fc=z-v\}\\\notag
&=\min\{\|c+u-u\|_1: z=F(c+u), v=Fu\}\\\label{equ:z-v}
&=\min\{\|x-u\|_1: z=Fx, v=Fu\},
\end{align}
which implies
\begin{equation}\label{equ:tW}
\sigma_{\Fc,s}(z)=\min_{v\in\Sigma_{\Fc,s}}\|z-v\|_{\Fc}=\min\{\|x - u\|_1, z = Fx, \ Fu\in\Sigma_{\Fc,s}\}.
\end{equation}
There has been many works on this case \cite{RSV08, CENR11,LML12,ACP12,CWW14,CCL18}. See Section \ref{sec:frame} for more details.
\end{example}

\begin{example}[Low Rank Matrix Recovery]\label{exa:matrix}
Let $\Mc=\{uv^T: u\in\R^{n_1}, v\in\R^{n_2}, \|u\|_2=\|v\|_2=1\}$. For any matrix $X\in\R^{n_1\times n_2}$, its atomic norm becomes the nuclear norm. Although $\|X\|_\Mc=\|X\|_*$ is a well known fact, we provide a proof in the Appendix (Lemma \ref{lem:nuc}). $\Sigma_{\Mc,s}=\{\sum_{i=1}^sc_iu_iv_i^T: \|u\|_2=\|v\|_2=1\}$ is the set of $n_1\times n_2$ matrices whose rank is at most~$s$. Let $K=\min\{n_1,n_2\}$. Note that the tail of any matrix $Z$ has a simple expression
\begin{equation}
\sigma_{\Mc,s}(Z)=\sum_{j=s+1}^K\sigma_j(Z).
\end{equation}
This is due to Lemma \ref{lem:sv}. Specifically, for $V\in\Sigma_{\Mc,s}$,  we have $\|Z-V\|_*=\sum_{i=1}^K\sigma_i(Z-V)\geq\sum_{i=1}^K|\sigma_i(Z)-\sigma_i(V)|=\sum_{i=1}^s|\sigma_i(Z)-\sigma_i(V)|+\sum_{i=s+1}^K\sigma_i(Z)\geq \sum_{i=s+1}^K\sigma_i(Z)$. So $\sigma_{\Mc,s}(Z)=\min_{V\in\Sigma_{\Mc,s}}\|Z-V\|_*=\sum_{j=s+1}^{K}\sigma_j(Z)$. 
Low rank matrix recovery through nuclear norm minimization have been investigated in \cite{RXH08, RXH11, KKRT16}. See Section \ref{sec:matrix} for more details.
\end{example}

\begin{example}[Phase Retrieval]\label{exa:p}
Let $\Pc=\{uu^T: u\in\C^{n},  \|u\|_2=1\}$ be the collection of unit norm rank-1's  in $\Hc_n$, the space of complex Hermitian $n\times n$ matrices. $\Hc_n$ is a real vector space of dimension $n^2$. For any Hermitian matrix $X\in\Hc_n$, its atomic norm is its nuclear norm. $\Sigma_{\Pc,s}=\{\sum_{i=1}^sc_iu_iu_i^*: \|u_i\|_2=1\}$ is the set of $n\times n$ Hermitian matrices whose rank is at most~$s$.

The phase retrieval problem aims to recover $z_0\in\C^n$ up to a phase from $\{|\langle z_0,a_i\rangle|\}_{i=1}^m$, which is equivalent to recovering $Z_0=z_0z_0^*$ from the linear measurement $\Ac (Z_0)=y:=(\langle Z_0,a_ia_i^*\rangle)_{i=1}^m$. One may solve (P$_{\Pc,\epsilon}$), which is
\begin{equation}\label{equ:P}
\hat{Z}=\argmin_{Z\in \Hc_n} \|Z\|_{*} \quad\text{ subject to } \quad \|\Ac(Z)-y\|_2\leq \epsilon.
\end{equation}
to recover $Z_0\in\Sigma_{\Pc,1}$, hence recovering $z_0$ up to a phase. See \cite{CL14, KRT17}.
\end{example}

\subsection{Contributions and Organizations}
Our work has a very general setting for recovering signals sparse in an atomic set via the convex programming \eqref{equ:K}. Many important applications fall into this general setting as outlined above. 
A major piece of ingredient in our work is the tail definition \eqref{equ:tWW}, which enables us to properly define the null space property and its variations related to programing \eqref{equ:K0} or \eqref{equ:K} in Section \ref{sec:recover}. This tail notion is simple yet effective, allowing intuitive and simplified analysis. Despite  its seemingly triviality, the author has not seen this tail formulation \eqref{equ:tWW} elsewhere except in \cite{F16}, which did not pursue a null space property based approach.

The first contribution is to provide thorough stability results for recovering $\Wc$-sparse signals using the programing \eqref{equ:K0} or \eqref{equ:K}, given the sensing map satisfies some null space property. Theorem \ref{thm:iff} characterizes the exact recovery of sparse signals via \eqref{equ:K0} when no noise is present. In order to cope with approximately sparse signal or noisy measurements, we propose a splittable condition (Definition \ref{def:split}) on the atomic set $\Wc$, which is interesting in its own right. Under this splittable condition, Theorem \ref{thm:main} and Theorem \ref{thm:RNSP} shows recovery through \eqref{equ:K} is stable and robust if $\Ac$ satisfies a robust null space property.

The second contribution is to derive simpler analysis and new results for recovering frame sparse signals and low rank matrices. 
Corollary \ref{cor:frame} provides new recovery results for frame sparse signals under the stable or robust null space property. 
This is particularly noteworthy as for recovering frame sparse signals, the $\ell_1$ synthesis method (see \eqref{equ:l1s}) has posed challenges in stability analysis. A result such as Corollary \ref{cor:frame} is desirable but was lacking. The general atomic setting allows a proper stable or robust null space property (see \eqref{equ:fstnsp} and \eqref{equ:frnsp}) and induces the new concept of splittability (Definition \ref{def:split}) which is crucial in the technical arguments.
\color{black}
Moreover, Corollary \ref{cor:frame2} is an improvement over \cite[Theorem 5.2]{CWW14}, as explained in Remark  \ref{rem:snsp}.
Corollary \ref{cor:fmin}  provides the minimum number of measurements needed using the $\ell_1$ synthesis method. 
 Corollary \ref{cor:matrix}(a)(c) are  new results for low rank matrix recovery.

Our work has a similar setup as \cite{C12}, but it differs from \cite{C12} in several ways. First, the recovery results in \cite{C12} focus on signals being exactly sparse whereas our settings are more robust. In fact, we focuses on the tail quantity which is more meaningful for non-sparse signals. Secondly, we emphasize on the null space property which allows us to provide deterministic recovery results. Consequently, any probabilistic  statements are uniform on the signals in the sense that the sensing map satisfies certain condition with high probability, and hence all signals can be recovered robustly. Thirdly, regarding the number of measurements, we also provide the minimum number of measurements required when using the convex programming \eqref{equ:K0} (Theorem \ref{thm:min2}).



The rest of the sections are organized as follows. Section \ref{sec:recover} presents the recovery results under null space properties for the general framework.
Section \ref{sec:min} investigates the minimum number of measurements required for successful recovery through \eqref{equ:K} or any stable decoder, given the atomic set $\Wc$ is finite and properly conditioned. Section \ref{sec:sub} focuses on the subgaussian measurements, which in many cases achieves the minimum number of measurements.
Sections \ref{sec:frame} and \ref{sec:matrix} focus on two special cases, signals sparse in a frame and low rank matrix recovery. 

\section{Recovery Results}\label{sec:recover}
In this section, we will provide stability and robustness results of recovering signals 
 through the minimization problem \eqref{equ:K0} or \eqref{equ:K} if the null space of the linear map $\Ac$ satisfies certain properties.

Given the set $\Wc$, and the sparsity level $s$, we define the set
\begin{equation}\label{equ:E}
E_{\Wc,s}:=\{z:\|v\|_\Wc<\|z-v\|_\Wc,  \text{ for any }v\in\Sigma_{\Wc,s}\}.
\end{equation}
Loosely speaking, the signals in $E_{\Wc,s}$ should have their atomic energy spread out over the atoms of $\Wc$. On the other hand, a $\Wc$-$s$-sparse vector has its energy concentrated at $s$ atoms only and obviously does not belong to $E_{\Wc,s}$. Furthermore,  one can easily show  that  $\Sigma_{\Wc,2s}\cap E_{\Wc,s}=\emptyset$.

Since the goal is to tell $\Wc$-$s$-sparse vectors apart after mapping through $\Ac$, at minimum, we need to have $\Nc(\Ac)\cap \Sigma_{\Wc,2s}=\{0\}$. The following property of $\Ac$ can guarantee this.
\begin{definition}\label{def:wnsp}
Given the set $\Wc$, and the sparsity level $s$, $\Ac$ is said to have $\Wc$ null space property of order $s$ ($\Wc$-$s$-NSP) if $\Nc(\Ac)\backslash\{0\}\subseteq E_{\Wc,s}$. Specifically, $\Wc$-$s$-NSP is equivalent to
\begin{equation}\label{equ:wnsp}
\|v\|_\Wc<\|z-v\|_\Wc, \quad  \text{ for any }z\in\Nc(\Ac)\backslash\{0\} \text{ and any }v\in\Sigma_{\Wc,s}.
\end{equation}
\end{definition}

The condition $\Wc$-$s$-NSP is clearly stronger than the bare minimum $\Nc(A)\cap \Sigma_{\Wc,2s}=\{0\}$, but it is also a necessary condition to recover all $\Wc$-$s$-sparse signals if we choose to use the minimization problem \eqref{equ:K0}.

\begin{theorem}\label{thm:iff}
$\Ac$ has $\Wc$-$s$-NSP  if and only if the method \eqref{equ:K0} is successful at recovering all signals in $\Sigma_{\Wc,s}$.
\end{theorem}
\begin{proof}
Suppose \eqref{equ:K0} is successful at recovering all signals in $\Sigma_{\Wc,s}$. Take any $v\in\Sigma_{\Wc,s}$ and any $z\in\Nc (\Ac)\backslash\{0\}$. We will use \eqref{equ:K0} to recover $-v$ (with $y=-\Ac v$). Both $-v$ and $z-v$ are feasible in \eqref{equ:K0} and $-v\neq z-v$. By assumption, $-v$ must be the unique minimizer, which means $\|-v\|_\Wc<\|z-v\|_\Wc$.

On the other hand, we assume $\Ac$ satisfies $\Wc$-$s$-NSP. Fix an arbitrary $v\in \Sigma_{\Wc,s}$, we solve  
$\min_{z\in \R^d} \|z\|_{\Wc}$  subject to $\Ac z=\Ac(v)$. For every feasible $z\neq v$, we have $v-z\in\Nc(\Ac)\backslash\{0\}$, so by $\Wc$-$s$-NSP, $\|v\|_\Wc<\|v-z-v\|_\Wc=\|z\|_\Wc$, which shows that the unique minimizer of this problem must be $v$, the signal to be recovered.
\end{proof}



\begin{remark}
In the setting of Example \ref{exa:cs}, Definition \ref{def:wnsp} reduces to the well known null space property of a matrix $A$, which is:
\begin{equation}\label{equ:nsp}
\|z_T\|_1<\|z_{T^c}\|_1, \text{ for any }z\in\Nc(A)\backslash\{0\} \text{ and any }|T|\leq s.
\end{equation}
This is because for any $z\in\Nc(A)\backslash\{0\}$ and any $v\in\Sigma_{\Ec,s}$, let $T$ be the support of $v$, then
$$\|v\|_1-\|z-v\|_1=\|v\|_1-\|z_T-v\|_1-\|z_{T^c}\|_1\leq\|v+z_T-v\|_1-\|z_{T^c}\|_1=\|z_T\|_1-\|z_{T^c}\|_1<0.$$
\end{remark}

The inequality $\|z_T\|_1<\|z_{T^c}\|_1$ is equivalent to $\frac{\|z_T\|_1}{\|z\|_1}<0.5$. Since the quantity $\frac{\|z_T\|_1}{\|z\|_1}$ is scaling free, for a fixed $T$, we may restrict $z$ in the compact set $\Nc(A)\cap \Sbd$.  Therefore $\frac{\|z_T\|_1}{\|z\|_1}$ obtain its minimum, which is strictly less than 0.5, on this compact set.
In summary, the NSP \eqref{equ:nsp} implies the existence of $\rho<1$ such that $\|z_T\|_1\leq\rho\|z_{T^c}\|_1$ for any $z\in\Nc(A)$ and any $|T|\leq s$. However, this is certainly not the case for Definition \ref{def:wnsp}. In fact, for any $\rho<1$, we can never have 
\begin{equation}\label{equ:wrong}
\|v\|_\Wc\leq\rho\|z-v\|_\Wc, \qquad  \text{ for any }z\in\Nc(\Ac), \text{ and any }v\in\Sigma_{\Wc,s}.
\end{equation}
 This is because with a fixed $z$, we can choose $v\in\Sigma_{\Wc,s}$ (by scaling) such that $\|v\|_\Wc$ is big enough to make $\|v\|_\Wc>\rho\|z-v\|_\Wc$.

In order to make \eqref{equ:wrong} feasible, we have to restrict the choice of $v$.
Therefore we will propose  two strengthened  null space properties that restrict $v$ to be the best $s$-term approximation of $z$.

\begin{definition}\label{def:wnsp2}
Recall $z_s$ as defined in \eqref{equ:zs}.

(a) $\Ac$ is said to have the stable $\Wc$ null space property of order $s$ with the NSP constant $\rho<1$ ($\Wc$-$s$-$\rho$-NSP) if  
\begin{equation}
\|z_s\|_\Wc\leq\rho\|z-z_s\|_\Wc, \qquad  \text{ for any }z\in\Nc(\Ac).
\end{equation}

(b) $\Ac$ is said to have the robust $\Wc$ null space property of order $s$ ($\Wc$-$s$-RNSP) if there exist  constants $\tau>0$ and $0<\rho<1$ such that 
\begin{equation}\label{equ:frnsp}
\|z_s\|_\Wc\leq\rho\|z-z_s\|_\Wc+\tau \|\Ac z\|_2, \qquad \text{for any }z\in\R^d.
\end{equation}
\end{definition}

We want to include Definition \ref{def:wnsp2}(b) since this type of ``null space'' property has been used in compressed sensing \cite{F14} and low matrix recovery \cite{KKRT16}.

We will see in Theorem \ref{thm:main} and Theorem \ref{thm:RNSP} that both versions of the strengthened null space property will imply the stable and robust recovery of any signal via \eqref{equ:K} if the atomic set $\Wc$ satisfies certain conditions. 
%
%
%

\begin{definition}\label{def:split}
We call $\Wc$ $s$-splittable with constant $\beta>0$ if for any $x,y\in\R^d$, 
\begin{equation}\label{equ:split}
\|x+y\|_\Wc\geq\|x_s\|_\Wc-\|y_s\|_\Wc+\beta(\|y-y_s\|_\Wc-\|x-x_s\|_\Wc).
\end{equation}
\end{definition}

Definition \ref{def:split} looks strange at first glance, but as two important examples, both $\Ec$ and $\Mc$ are $s$-splittable with $\beta=1$. See Lemma \ref{lem:split} in the Appendix. For a frame $\Fc$, \eqref{equ:split} is likely satisfied with a small enough $\beta<1$. Further investigation is needed but we do provide an example in Lemma \ref{lem:fsplit}.

\begin{theorem}\label{thm:main}
Let $\Wc$ be $s$-splittable with constant $\beta_{\Wc}$. If $\Ac$ has the stable $\Wc$-NSP of order $s$ with  constant $0<\rho<\beta_\Wc$, then for any $z_0\in\R^d$ and $y=\Ac z_0+e$ where $\|e\|_2\leq\epsilon$, we have

\begin{equation}\label{equ:main}
\|\hat z-z_0\|_\Wc\leq\frac{(1+\beta_\Wc)(1+\rho)}{\beta_\Wc(\beta_\Wc-\rho)}\sigma_{\Wc,s}(z_0)+\frac{2C_\Wc(\beta_\Wc+1)}{\nu_\Ac(\beta_\Wc-\rho)}\epsilon,
\end{equation}

where $\hat z$ is from \eqref{equ:K}. 
\end{theorem}
\begin{proof}
Let $h=\hat z-z_0$. 
We decompose $h$ as $h=a+\eta$ where $a\in \Nc(\Ac)$ and $\eta\in \Nc(\Ac)^\perp$. We have $$\|\eta\|_2\leq\frac{1}{\nu_\Ac}\|\Ac h\|_2\leq\frac{2\epsilon}{\nu_\Ac}.$$

On one hand, using that $\hat z$ is a minimizer and $\Wc$ is splittable,
\begin{align*}
\|z_0\|_\Wc&\geq\|\hat z\|_\Wc=\|a+\eta+z_0\|_\Wc\geq\|z_0+a\|_\Wc-\|\eta\|_\Wc\\
&\geq\|z_{0,s}\|_\Wc-\|a_s\|_\Wc+\beta_\Wc\left(\|a-a_s\|_\Wc-\|z_0-z_{0,s}\|_\Wc\right)-\|\eta\|_\Wc\\
&\geq\|z_0\|_\Wc-\|z_0-z_{0,s}\|_\Wc-\|a_s\|_\Wc+\beta_\Wc\|a-a_s\|_\Wc-\beta_\Wc\|z_0-z_{0,s}\|_\Wc-\|\eta\|_\Wc.
\end{align*}
If we denote $\|z_0-z_{0,s}\|_\Wc$ as $\sigma$, then the above simplifies to
\begin{equation}\label{equ:snsp1}
\beta_\Wc\|a-a_s\|_\Wc\leq(1+\beta_\Wc)\sigma+\|a_s\|_\Wc+\|\eta\|_\Wc
\end{equation}

On the other hand, by stable $\Wc$-NSP,
\begin{equation}\label{equ:a}
\|a_s\|_\Wc\leq\rho\|a-a_s\|_\Wc.
\end{equation}
Combining \eqref{equ:snsp1} and \eqref{equ:a}, we have
$$\frac{\beta_\Wc}{\rho}\|a_s\|_\Wc\leq(1+\beta_\Wc)\sigma+\|a_s\|_\Wc+\|\eta\|_\Wc,$$
 which simplifies to
\begin{equation}\label{equ:v0}
\|a_s\|_\Wc\leq\frac{(1+\beta_\Wc)\rho}{\beta_\Wc-\rho}\sigma+\frac{\rho}{\beta_\Wc-\rho}\|\eta\|_\Wc.
\end{equation}
In the end, using \eqref{equ:snsp1} and \eqref{equ:v0},
\begin{align*}
\|h\|_\Wc&\leq\|a\|_\Wc+\|\eta\|_\Wc\leq\|a_s\|_\Wc+\|a-a_s\|_\Wc+\|\eta\|_\Wc\\
&\leq\|a_s\|_\Wc+\frac{1+\beta_\Wc}{\beta_\Wc}\sigma+\frac{1}{\beta_\Wc}\|a_s\|_\Wc+\frac{1+\beta_\Wc}{\beta_\Wc}\|\eta\|_\Wc\\
&\leq\frac{1+\beta_\Wc}{\beta_\Wc}\left[\frac{(1+\beta_\Wc)\rho}{\beta_\Wc-\rho}\sigma+\frac{\rho}{\beta_\Wc-\rho}\|\eta\|_\Wc\right]+\frac{1+\beta_\Wc}{\beta_\Wc}(\sigma+\|\eta\|_\Wc)\\
&=\frac{(1+\beta_\Wc)(1+\rho)}{\beta_\Wc(\beta_\Wc-\rho)}\sigma+\frac{\beta_\Wc+1}{\beta_\Wc-\rho}\|\eta\|_\Wc\leq\frac{(1+\beta_\Wc)(1+\rho)}{\beta_\Wc(\beta_\Wc-\rho)}\sigma+\frac{\beta_\Wc+1}{\beta_\Wc-\rho}C_\Wc\|\eta\|_2\\
&\leq\frac{(1+\beta_\Wc)(1+\rho)}{\beta_\Wc(\beta_\Wc-\rho)}\sigma+\frac{2C_\Wc(\beta_\Wc+1)}{\nu_\Ac(\beta_\Wc-\rho)}\epsilon.
\end{align*}
\end{proof}

\begin{theorem}\label{thm:RNSP}
Let $\Wc$ be $s$-splittable with constant $\beta_{\Wc}$. Given $z_0\in\R^d$ and the recovered signal $\hat z$ from \eqref{equ:K} where $y=\Ac z_0+e$ with $\|e\|_2\leq\epsilon$, then the robust $\Wc$-NSP of order $s$ with  constant $0<\rho<\beta_\Wc$ and $\tau>0$ yields the following stability result
\begin{equation}
\|\hat z-z_0\|_\Wc\leq\frac{(1+\beta_\Wc)(1+\rho)}{\beta_\Wc-\rho}\sigma_{\Wc,s}(z_0)+\frac{2(1+\beta_\Wc)\tau}{\beta_\Wc-\rho}\epsilon.
\end{equation}
\end{theorem}

\begin{proof}
Let $h=\hat z-z_0$. Denote $\sigma_{\Wc,s}(z_0)=\|z_0-z_{0,s}\|_\Wc$ by $\sigma$.
Using $\hat z$ is a minimizer and $\Wc$ is splittable,
\begin{align*}
\|z_0\|_\Wc&\geq\|\hat z\|_\Wc=\|z_0+h\|_\Wc\geq\|z_{0,s}\|_\Wc-\|h_s\|_\Wc+\beta_\Wc(\|h-h_s\|_\Wc-\sigma)\\
&\geq\|z_0\|_\Wc-\sigma-\|h_s\|_\Wc+\beta_\Wc\|h-h_s\|_\Wc-\beta_\Wc\sigma,
\end{align*}
which simplifies to
\begin{equation}\label{equ:rnsp1}
\beta_\Wc\|h-h_s\|_\Wc\leq\|h_s\|_\Wc+(1+\beta_\Wc)\sigma.
\end{equation}
Now we use the $\Wc$-RNSP to get
\begin{equation}\label{equ:rnsp2}
\|h_s\|_\Wc\leq\rho\|h-h_s\|_\Wc+\tau\|\Ac h\|_2.
\end{equation}
\eqref{equ:rnsp1} and \eqref{equ:rnsp2} implies
\begin{equation}\label{equ:main3}
\|h_s\|_\Wc\leq\frac{1}{\beta_\Wc-\rho}((1+\beta_\Wc)\rho\sigma+\beta_\Wc\tau\|\Ac h\|_2).
\end{equation}
Using \eqref{equ:rnsp1} and \eqref{equ:main3},
\begin{align*}
\|h\|_\Wc&\leq\|h_s\|_\Wc+\|h-h_s\|_\Wc\leq\frac{1+\beta_\Wc}{\beta_\Wc}(\|h_s\|_\Wc+\sigma)\\
&\leq\frac{1+\beta_\Wc}{\beta_\Wc}\left(\frac{1}{\beta_\Wc-\rho}((1+\beta_\Wc)\rho\sigma+\beta_\Wc\tau\|\Ac h\|_2)+\sigma\right)\\
&=\frac{(1+\beta_\Wc)(1+\rho)}{\beta_\Wc-\rho}\sigma+\frac{1+\beta_\Wc}{\beta_\Wc-\rho}\tau\|\Ac h\|_2.
\end{align*}
Finally, the desired result holds because $\|\Ac h\|_2\leq 2\epsilon$.
\end{proof}

Finally, we proposed a third kind  of null space property.


\begin{definition}\label{def:snsp}
$\Ac$ is said to have the strong $\Wc$ null space property of order $s$ ($\Wc$-$s$-SNSP) if  there exists $c>0$ such that
\begin{equation}\label{equ:snsp}
\|z-v\|_\Wc-\|v\|_\Wc\geq c\|z\|_2, \qquad  \text{ for any }z\in\Nc(\Ac), \text{ and any }v\in\Sigma_{\Wc,s}.
\end{equation}
\end{definition}

 As we will see in the following theorem, recovery results with this property does not pose further conditions on $\Wc$. Definition \ref{def:snsp} is inspired by \cite{CWW14} and will be discussed more in Section \ref{sec:frame}. The use of $\|\cdot\|_2$ in \eqref{equ:snsp} is essentially due to the same norm used in measurement noise in \eqref{equ:K}.

\begin{theorem}\label{thm:snsp}
If $\Ac$ has the strong $\Wc$-NSP of order $s$ as in \eqref{equ:snsp}, then for any $z_0\in\R^d$ and $y=\Ac z_0+e$ where $\|e\|_2\leq\epsilon$, we have
$$\|\hat z-z_0\|_2\leq\frac{2}{\nu_\Ac}\left(\frac{C_\Wc}{c}+1\right)\epsilon+\frac{2}{c}\sigma_{\Wc,s}(z_0),$$
where $\hat z$ is from \eqref{equ:K}. 
\end{theorem}
\begin{proof}

Let $\sigma_{\Wc,s}(z_0)=\|z_0-v_0\|_{\Wc}$, where $v_0\in\Sigma_{\Wc,s}.$ 
 Let $h=\hat z-z_0$. We abbreviate $\sigma_{\Wc,s}(z_0)$ to $\sigma$, then
 $\|v_0\|_\Wc+\sigma\geq\|z_0\|_\Wc\geq\|\hat z\|_\Wc=\|h+z_0\|_\Wc=\|h+v_0+z_0-v_0\|_\Wc\geq\|h+v_0\|_\Wc-\sigma$, which simplifies to
\begin{equation}\label{equ:main1}
\|h+v_0\|_\Wc\leq\|v_0\|_\Wc+2\sigma.
\end{equation}

On the other hand, we decompose $h$ as $h=a+\eta$ where $a\in \Nc(\Ac)$ and $\eta\in \Nc(\Ac)^\perp$. We have $$\|\eta\|_2\leq\frac{1}{\nu_\Ac}\|\Ac h\|_2\leq\frac{2\epsilon}{\nu_\Ac}.$$

By strong NSP,
\begin{equation}\label{equ:a4}
\|a+v_0\|_\Wc-\|v_0\|_\Wc\geq c\|a\|_2.
\end{equation}
Combining \eqref{equ:main1} and \eqref{equ:a4}, we have
\begin{align*}
c\|a\|_2&\leq\|a+v_0\|_\Wc-\|h+v_0\|_\Wc+2\sigma_{\Wc,s}(z_0)\\
&\leq\|a-h\|_\Wc+2\sigma_{\Wc,s}(z_0)=\|\eta\|_\Wc+2\sigma_{\Wc,s}(z_0).
\end{align*}
In the end,
\begin{align*}
\|h\|_2&\leq\|a\|_2+\|\eta\|_2\leq\frac{1}{c}\|\eta\|_\Wc+\frac{2}{c}\sigma_{\Wc,s}(z_0)+\|\eta\|_2\\
&\leq\left(\frac{C_\Wc}{c}+1\right)\|\eta\|_2+\frac{2}{c}\sigma_{\Wc,s}(z_0)\leq\frac{2}{\nu_\Ac}\left(\frac{C_\Wc}{c}+1\right)\epsilon+\frac{2}{c}\sigma_{\Wc,s}(z_0)
\end{align*}
\end{proof}

\section{Minimum Number of Measurements}\label{sec:min}
This section provides results on the smallest number of measurements required for successful recovery, whether it is through  \eqref{equ:K}, or another general decoder $\Delta$.
In this section, we will assume that there are finitely many atoms in $\Wc$ and impose the following $s$-even condition. Therefore the definitions and results in this section does not apply to $\Mc$ (Example \ref{exa:matrix}) or $\Pc$ (Example \ref{exa:p}).

\begin{definition}
The atomic set $\Wc=\{w_i\}_{i=1}^N$ is said to be $s$-even if for any index set $|J|\leq s$,  and any scalar $|c_j|=1$, we have $\|\sum_{j\in J}c_jw_j\|_\Wc=\sum_{j\in J}|c_j|=|J|$. 
\end{definition}

To understand the above definition better, we see that $\|\sum_{j\in J}c_jw_j\|_\Wc=|J|$ is equivalent to $\sum_{j\in J}\frac{c_j}{|J|}w_j$ being on the boundary of $\text{conv}(\Wc_{\text{sym}})$. On the convex geometry level, this is requiring every  subset of $s$ (or less) extreme points (the set $\Wc_{\text{sym}}$) to \textbf{be on} a face of the polytope. If we require every  subset of $s$ extreme points to \textbf{form} a face of the polytope, this is the so called  \emph{$s$-neighborly} polytope~\cite{BLN13}.The remark below argues that the $s$-even property is not too harsh given $s$ is small.

\begin{remark}
In the context of Example \ref{exa:frame}, a sufficient condition for $\Fc$ to be $s$-even is for the matrix $F=[f_1,f_2,\cdots,f_n]$ to have the $s$-NSP condition. A proof is provided in the Appendix (Lemma \ref{lem:snsp}). The interested readers can also refer to \cite{D04} for a similar result.
\end{remark}

 
\begin{lemma}[{\cite[Lemma 2.3]{FPRT10}}]\label{lem:N}
Suppose $t<N$ are integers. There exists a family $\Uc$ of subsets of $[N]$ such that:
\begin{enumerate}
\item Every set in $\Uc$ consists of exactly $t$ elements.
\item For all $I,J\in\Uc$ with $I\neq J$, it holds that $|I\cap J|<\frac{t}{2}$.
\item $|\Uc|\geq\left(\frac{N}{4t}\right)^{t/2}$.
\end{enumerate}
\end{lemma}

The above lemma is a famous combinatorial result, which is essential in proving both Theorem \ref{thm:min2} and Theorem \ref{thm:min}. Theorem \ref{thm:min2} focuses on the case where we have exact recovery of sparse signals from its precise measurement (no noise) via \eqref{equ:K0}. Note that the number of measurements can be on the order of $s$ if one solves $y=\Ac z_0$ via an exhaustive search. If we use the much more efficient convex programming \eqref{equ:K0}, we have to increase the number of measurements, but not by much. To prove Theorem \ref{thm:min2}, the main idea is that $\Ac$ needs to have enough measurements to be able to distinguish two different sparse vectors via minimizing the atomic norm, which is reflected in that a large collection of disjoint unit atomic balls remain disjoint under the mapping $\Ac$.

\begin{theorem}\label{thm:min2}
  Let $\Wc=\{w_i\}_{i=1}^N$ be $s$-even. If for every $z_0\in\Sigma_{\Wc,s}$ and $y=\Ac z_0$, we have $\hat z=z_0$ where $\hat z$ is a solution of \eqref{equ:K0}, then the number of measurments $m$ satisfies
$$m\geq \frac{1}{8\ln3}s\ln\frac{N}{2s}.$$
\end{theorem}
\begin{proof}
Let $\Uc$ be a family of subsets of $[N]$ given by Lemma \ref{lem:N} with $t=\lfloor s/2\rfloor$. For each $I\in\Uc$, define $z_I:=\frac{1}{t}\sum_{i\in I}w_i$. We have $|I|=t$. By definition, $\|z_I\|_\Wc\leq1$.

For $I\neq J$,  $z_I-z_J=\frac{1}{t}(\sum_{i\in I-J}w_i-\sum_{j\in J-I}w_j)$. We have $|I-J|+|J-I|=|I|+|J|-2|I\cap J|$, so $s\geq2t\geq|I-J|+|J-I|>t$ by Lemma \ref{lem:N}. Since $\Wc$ is $s$-even, 
\begin{equation}\label{equ:zij}
\|z_i-z_j\|_\Wc=\frac{1}{t}(|I-J|+|J-I|)>1.
\end{equation}

We claim that $\{\Ac(z_I+\frac{1}{2} B_\Wc), I\in\Uc\}$ is a disjoint collection of subsets of $\Ac(\R^d)$.
 Otherwise there exist $I\neq J$ and $v, v'\in B_\Wc$ such that $\Ac(z_I+\frac{1}{2} v)=\Ac(z_J+\frac{1}{2} v')$. This means that $\Ac(z_I-z_J)=\Ac(\frac{1}{2} v'-\frac{1}{2} v)$. The vector $z_I-z_J\in\Sigma_{\Wc,s}$, so by our assumption,  
 \begin{equation}
 \|z_I-z_J\|_\Wc\leq\|\frac{1}{2} v'-\frac{1}{2} v\|_\Wc\leq\|\frac{1}{2} v'\|_\Wc+\|\frac{1}{2} v\|_\Wc=1,
 \end{equation}
a contradiction to \eqref{equ:zij}.

It is easy to see that $\Ac(z_I+\frac{1}{2} B_\Wc)\subset \frac{3}{2}\Ac B_\Wc$ for any $I\in\Uc$. So
\begin{align}\label{equ:vol2}
|\Uc|\vol(\frac{1}{2}\Ac B_\Wc)))\leq \vol(\frac{3}{2}\Ac B_\Wc).
\end{align}
Let $\Ac B_\Wc$ have  dimension (as a manifold) $r\leq m$, so \eqref{equ:vol2} becomes
\begin{align*}
\left(\frac{N}{4t}\right)^{t/2}(\frac{1}{2})^r\vol(\Ac B_\Wc)\leq(\frac{3}{2})^r\vol(\Ac B_\Wc)
\Longrightarrow  \left(\frac{N}{4t}\right)^{t/2}\leq3^r\leq3^m.
\end{align*}
Taking log of both sides yields
$$m\ln3\geq\frac{t}{2}\ln\frac{N}{4t}\geq\frac{s/2-1}{2}\ln\frac{N}{2s}\geq\frac{s}{8}\ln\frac{N}{2s},$$ if $s>2$.

\end{proof}

The next theorem is the minimum number of measurements required for the existence of a stable decoder, not necessarily the decoder \eqref{equ:K} or \eqref{equ:K0}. The proof is in the Appendix since it's quite similar to that of Theorem \ref{thm:min2}.
 \begin{theorem}\label{thm:min}
 Let $\Wc$ be $s$-even. Suppose that there exists a linear map $\Ac$ and a reconstruction map $\Delta$ stable in the sense that
 \begin{equation}\label{equ:stable}
 \|z-\Delta(\Ac z)\|_\Wc\leq C\sigma_{\Wc,s}(z),\quad\text{ for any }z\in\R^d,
 \end{equation}
 then there exists $C'>0$ depending only on $C$ such that the number of measurments
 $$m\geq C's\ln(eN/s).$$
 \end{theorem}

In the case of compressed sensing, both theorems can be found in  \cite{FRbook}. For example, Theorem \ref{thm:min2} is reduced to \cite[Theorem 10.11]{FRbook}.


\section{Subgaussian Measurements}\label{sec:sub}
In this section we show that $\Wc$-sparse vectors can be recovered from few random measurements via \eqref{equ:K}. The arguments presented here are similar to those in \cite{CCL18}, but the results are still worth stating for the general setting.

\begin{definition}
	The \emph{Gaussian width} of a set $S\subset\R^d$ is defined as
	\[
	w(S) := \E \sup_{x \in S} \langle g , x \rangle,
	\]
	where $g \sim N(0, I_d)$ is a standard Gaussian random vector.
\end{definition}

\begin{definition} A random vector $\vphi\in\R^d$ is called a \emph{subgaussian vector} with parameters $(\alpha, \sigma)$ if it satisfies the following.
	\begin{enumerate}[{(1)}]
		\item It is centered, that is, $\E [\vphi]= 0$.
		\item There exists a positive $\alpha$ such that $\E\left[|\langle\vphi, z\rangle|\right]\geq\alpha$ for every $z\in\Sbd$.
		\item There exists a positive $\sigma$ such that $\Pr\left(|\langle\vphi,z\rangle|\geq t\right)\leq 2\exp(-\frac{t^2}{2\sigma^2})$ for every $z\in\Sbd$.
	\end{enumerate}
\end{definition}

The proof of the following lemma can be found in \cite[Section 6.5]{T14}.

\begin{lemma}\label{lem:Q}
	If $\vf \in \R^d$ is a subgaussian vector with parameters $(\alpha, \sigma)$, then 
	\[
	\Pr\left[|\langle \vx, \vf \rangle|\geq t\right] \geq \frac{(\alpha-t)^2}{4\sigma^2}
	\]	
	for any $0<t<\alpha$ and \(\vx \in \Sbd\).
\end{lemma}

If $\{\vphi_i\}_{i=1}^q$ are independent copies of the random distribution $\vphi \in \R^d$, then we can define the \emph{mean empirical width} of a set $S \subset \R^d$ as
$$W_q(S; \vphi) := \E \sup_{x \in S} \bigg\langle x, \dfrac{1}{\sqrt{q}} \sum_{i = 1}^q \e_i \vphi_i \bigg\rangle,$$
	where $\{\e_i\}_{i = 1}^m$ are independent random variables taking values uniformly over $\{\pm 1\}$ and are independent from everything else. 

The mean empirical width $W_q(S;\vphi)$ is a distribution-dependent measure of the size of the set $S$. 
Note that $W_q(S; \vphi)$ reduces to the usual Gaussian width $w(S)$ when $\vphi$ follows a standard Gaussian distribution.
Estimation of \(W_m(S; \vphi)\) for any subgaussian vector \(\vphi\) is made in \cite{T14}, where \(S\) is required to be $\Sbd\cup G$ for some cone $G$. However, the bound can be relaxed to any subset \(S\) by the observation of the generic chaining bound and the majorizing measure theorem \cite[{Theorem 2.2.18 and Theorem 2.4.1}]{Ta14}. We will state this as a lemma.

\begin{lemma}\label{lem:W}
If $\vphi\in\R^d$ is subgaussian with  parameters $(\alpha,\sigma)$ and $S$ is any subset of $\R^d$, then
\begin{equation}
W_m(S;\vphi)\leq Cw(S)
\end{equation}
for some universal constant $C$.
\end{lemma}

The constant $C$ is a universal constant that does not rely on
the choice of subgaussian distribution. See \cite{Ta14} for precise computations of this constant, and \cite[Remark 2.7]{CCL18} for the computation of $C$ when $\vphi$ is multivariate normal.

The mean empirical width appears in the following important result. This theorem was originally stated in \cite{KM15} and coined as  \emph{Mendelson's Small Ball Method} by  Tropp \cite{T14}. This will be a primary tool in obtaining our main estimates.

\begin{theorem}[{\cite[Proposition 5.1]{T14}}, cf. {\cite[Theorem 2.1]{KM15}}]\label{thm:QW}
	Fix a set $S \subset \mathbb{R}^d$. Let $\vphi$ be a random vector in $\mathbb{R}^d$ and let $A \in \R^{m \times d}$ have rows $\{a_i^T\}_{i = 1}^m$ that are independent copies of $\vphi^T$. Define
	\[
	Q_\xi (S; \vphi) := \inf_{\vx \in S} \Pr\bigg(|\langle \vx, \vphi \rangle| \geq \xi \bigg).
	\]
	Then for any $\xi > 0$ and $t > 0$, we have
	\begin{align}\label{equ:msbineq}
	\inf_{\vx \in S} \| A \vx \|_2 \geq \xi \sqrt{m} Q_{2\xi}(S; \vphi) - 2 W_m(S;\vphi) - \xi t
	\end{align}
	with probability $\geq 1 - e^{-t^2/2}$.
\end{theorem}


Let $$S_\rho:=\{z:\|z_s\|_\Wc\geq\rho\|z-z_s\|_\Wc \}\cap\Sbd.$$
It is easy to show that $\inf_{\vx\in S_\rho}\|\Ac \vx\|_2>0$ implies $\Ac$ having the stable $\Wc$-$\rho$-NSP. We let $A$ be the matrix representation of the map $\Ac$ under an orthonormal basis.

\begin{theorem}
Assume \(\vphi \in \R^d \) is a subgaussian vector with parameters \((\alpha,\sigma) \). If $A\in \R^{m \times d}$ is a measurement matrix with rows that are independent copies of $\vphi^T$ and that the number of measurements satisfies 
	\[
	m \geq \frac{4^{8}\sigma^4}{\alpha^6}C^2w^2( S_\rho),
	\]
	then with probability at least  
	\[
	1 -\exp\left({- m\dfrac{\alpha^4}{64^2\sigma^4}}\right),
	\]
	we have  
	\begin{equation}\label{equ:infPhiD}
	\displaystyle\inf_{\vx\in S_\rho}\|\Ac \vx\|_2\geq Cw(S_\rho).
	\end{equation}

\end{theorem}
\begin{proof}
	We first apply Lemma \ref{lem:W} and Theorem \ref{thm:QW},  to obtain the bound
	\begin{equation}\label{equ:QW}
	\inf_{\vx \in S_\rho} \| A  \vx \|_2 \geq \xi \sqrt{m} Q_{2\xi}(S_\rho; \, \vphi) - 2 Cw(S_\rho)- \xi t.
	\end{equation}
	By Lemma \ref{lem:Q}, provided we choose $\xi < \alpha/2$, we obtain 
	\begin{align*}
	Q_{2\xi}(S_\rho; \vphi) 
	= \inf_{\vx\in S_\rho} \Pr\left(|\langle \vx, \vphi \rangle| \geq 2\xi \right)
	\geq \frac{(\alpha-2\xi)^2}{4\sigma^2}.
	\end{align*}
	
	
	Placing the bound for $Q_{2\xi}(S_\rho; \vphi)$ into \eqref{equ:QW} and choosing $\xi=\alpha/4$ gives 
\begin{equation}
\inf_{\vx \in S_\rho} \| A  \vx \|_2\geq\frac{\alpha^3\sqrt{m}}{4^3\sigma^2}-2 Cw(S_\rho)-\frac{\alpha t}{4}:=a-b-\frac{\alpha t}{4}
\end{equation}

	Picking \(m \) and \(t\) to satisfy \(a \geq 2b\) and \(\alpha t /4 = (a - b)/2 \) gives
	\begin{equation*}
	\inf_{x \in S_\rho} \| A  \vx \|_2 \geq
	a - b - (a -b)/2 = (a - b)/2 \geq b/2 = Cw(S_\rho).
	\end{equation*}
	All that is left is to rewrite these  conditions in terms of \(m\) and \(t\). We have 
	\[
	a \geq 2b 
	\quad \Longleftrightarrow \quad
	\frac{\alpha^3}{4^3\sigma^2}\sqrt{m}\geq 4 Cw(S_\rho)
	\quad \Longleftrightarrow \quad 
	m \geq \frac{4^{8}\sigma^4C^2}{\alpha^6}w^2(S_\rho)
	\]
	and
	\[
	\frac{\alpha}{4}t=\frac{a-b}{2}\geq\frac{a}{4}=\frac{\alpha}{4^4}\left(\frac{\alpha}{\sigma}\right)^2\sqrt{m} \quad \Longleftrightarrow \quad  t\geq \frac{1}{64}\sqrt{m}\left(\frac{\alpha}{\sigma}\right)^2 \quad \Longleftrightarrow \quad -\dfrac{t^2}{2} \leq - m\dfrac{\alpha^4}{64^2\sigma^4},
	\]
	proving the result.
\end{proof}

As mentioned earlier, \eqref{equ:infPhiD} holds implies $\Ac$ has  the stable $\Wc$-$\rho$-NSP, therefore we have the following corollary.

\begin{corollary}\label{cor:sub}
Let $\Wc$ be $s$-splittable with constant $\beta_\Wc>0$ and fix $\rho<\beta_\Wc$. If $A$ is a measurement matrix with rows that are independent copies of a subgaussian vector with parameters $(\alpha,\sigma)$ and $m\geq \frac{4^{8}\sigma^4}{\alpha^6}C^2w^2( S_\rho)$, then $A$ has  stable $\Wc$-$\rho$-NSP with probability at least $1 -\exp\left({- m\dfrac{\alpha^4}{64^2\sigma^4}}\right)$, and therefore \eqref{equ:main} holds.
\end{corollary}

The computation of Gaussian width is in general difficult. $w(S_\rho)$ has been computed in the compressed sensing case, which has an upper bound of $O(s\ln(d/s)$~\cite[Proposition 9.33]{FRbook}. It is also bounded in the low matrix recovery case in \cite{KKRT16}. For the frame case, the minimum number of subgaussian measurements was computed through bounding a different Gaussian width~\cite{CCL18}, instead of $w(S_\rho)$. 

\begin{remark}
We may define a set according to the strong null space property as $S'_c=\{\|z-v\|_\Wc-\|v\|_\Wc\leq c\|z\|_2, \text{ for some }v\in\Sigma_{\Wc,s}\}\cap\Sbd$. This way, we will have a similar result to Corollary \ref{cor:sub} without requiring $\Wc$ being splittable. However, computing $w(S'_c)$ could pose bigger challenge.
\end{remark}

%

%
%

\section{Recovering Vectors Sparse in a Frame}\label{sec:frame}

This section will focus on the special case  in Example \ref{exa:frame}.
Recall $\Fc=\{f_i\}_{i=1}^N\subset\R^d$, and $F=[f_1,f_2,\cdots,f_N]$ is the $d\times N$ matrix whose columns are the atoms. 
The linear operator $\Ac$ will be represented by the matrix $A$. The minimization problem \eqref{equ:K} becomes
\begin{equation}\label{equ:F}\tag{$\text{P}_{\Fc,\epsilon}$}
\hat{z}=\arg\min_{z\in \R^d} \|z\|_{\Fc} \quad\text{ subject to } \quad \|A z-y\|_2\leq \epsilon.\\
\end{equation}

We will prove first that \eqref{equ:F} is the same as the  $\ell_1$-synthesis method~\cite{RSV08, CWW14, CCL18}, which is defined as
\begin{equation}\label{equ:l1s}
\left\{\begin{array}{l}
\hat{x}=\arg\min_{x\in \R^n} \|x\|_1 \quad\text{ subject to } \quad \|AFx-y\|_2\leq \epsilon,\\
\hat z = F\hat x
\end{array}\right.
\end{equation}

\begin{lemma} 
The problem \eqref{equ:F} is equivalent to \eqref{equ:l1s}.
\end{lemma}
\begin{proof}
Let $\bar z=F\bar x$ be a solution of \eqref{equ:F} and $\|\bar z\|_\Fc=\|\bar x\|_1$. For an arbitrary feasible vector $x$ of \eqref{equ:l1s}, $Fx$ is feasible in \eqref{equ:F}, so we have $\|\bar x\|_1=\|\bar z\|_\Fc\leq\|Fx\|_\Fc\leq\|x\|_1$. This makes $\bar z$ a solution of \eqref{equ:l1s}.

On the other hand, let $\hat z=F\hat x$ be a solution of \eqref{equ:l1s}. For an arbitrary feasible vector $z$ of \eqref{equ:F}, let $z=Fx$ and $\|z\|_\Fc=\|x\|_1$. We have that $x$ is feasible in \eqref{equ:l1s}, so $\|\hat z\|_\Fc\leq\|\hat x\|_1\leq\|x\|_1=\|z\|_\Fc$, making $\hat z$ a solution of \eqref{equ:F}.
\end{proof}

It is worthwhile to restate Definition \ref{def:wnsp} for the frame case: 
\begin{equation}\label{equ:FNSP}
\text{For any }z\in\Nc(A)\backslash\{0\} \text{ and any }v\in\Sigma_{F,s}, \text{ we have }\|v\|_\Fc<\|z-v\|_\Fc.
\end{equation}  

A direct consequence of Theorem \ref{thm:iff} is the following:
\begin{corollary}\label{cor:FNSP}
 The matrix $A$ satisfying \eqref{equ:FNSP} is the necessary and sufficient condition for the successful recovery of all signals in $\Sigma_{F,s}$ via \eqref{equ:F} or \eqref{equ:l1s} ($\epsilon=0$).
\end{corollary}
Corollary \ref{cor:FNSP} is the same as  \cite[Theorem 4.2]{CWW14} since the dictionary NSP property proposed in \cite[Definition 4.1]{CWW14} also characterizes the $\ell_1$-synthesis method, and therefore is equivalent to \eqref{equ:FNSP}. 
A direct proof of the equivalence of \eqref{equ:FNSP} and \cite[Definition 4.1]{CWW14}, using only their definitions, will be omitted here as it is very similar to the proof of Lemma \ref{lem:fsnsp}.
We will also avoid stating \cite[Definition 4.1]{CWW14} here as it is the same as \eqref{equ:fsnsp2} when setting $c=0$. As seen, \eqref{equ:FNSP} is a more concise version than \cite[Definition 4.1]{CWW14}, which is due to the tail definition \eqref{equ:tWW}. This has also created simpler proof for  \cite[Theorem 4.2]{CWW14}.




We  combine  Theorem \ref{thm:main} and Theorem \ref{thm:RNSP} for the frame case as one corollary.

\begin{corollary}\label{cor:frame}
Suppose $\Fc$ is $s$-splittable with constant $\beta$.

(a) If $A$ has the stable $\Fc$ null space property of order $s$ with $\rho <\beta$ as 
\begin{equation} \label{equ:fstnsp}
\|z_s\|_\Fc\leq\rho\|z-z_s\|_\Fc, \qquad  \text{ for any }z\in\Nc(A),
\end{equation}
then for any $z_0\in\R^d$ and $y=A z_0+e$ where $\|e\|_2\leq\epsilon$, we have
$$\|\hat z-z_0\|_\Fc\leq\frac{(1+\beta)(1+\rho)}{\beta(\beta-\rho)}\sigma_{\Fc,s}(z_0)+\frac{2\sqrt{N}(1+\beta)}{(\beta-\rho)\nu_A}\epsilon,$$
where $\hat z$ is from \eqref{equ:l1s}.

(b) If $A$ has the robust $\Fc$ null space property of order $s$ with constant $\rho<\beta, \tau>0$ as
\begin{equation}  \label{equ:frnsp}
\|z_s\|_\Fc\leq\rho\|z-z_s\|_\Fc+\tau \|A z\|_2, \qquad \text{for any }z\in\R^d, 
\end{equation}
then for any $z_0\in\R^d$ and $y=A z_0+e$ where $\|e\|_2\leq\epsilon$, we have
\begin{equation*}
\|\hat z-z_0\|_\Fc\leq\frac{(1+\beta)(1+\rho)}{\beta(\beta-\rho)}\sigma_{\Fc,s}(z_0)+\frac{2\tau(1+\beta)}{\beta-\rho}\epsilon,
\end{equation*}
where $\hat z$ is from \eqref{equ:l1s}. 
\end{corollary}

Further investigation is needed on when the frame $\Fc$ is splittable. As a simple example, the frame $\{(\cos\frac{2\pi (n-1)}{8},\sin\frac{2\pi (n-1)}{8})\}_{n=1}^8$ of $\R^2$ is 1-splittable with the constant $\sqrt{2}-1$. See Lemma \ref{lem:fsplit} in the Appendix.

Theorem \ref{thm:snsp} reduces to the following corollary:
\begin{corollary}\label{cor:frame2}
If $A$ has the strong $\Fc$ null space property of order $s$ as
\begin{equation}\label{equ:fsnsp}
\|z-v\|_\Fc-\|v\|_\Fc\geq c\|z\|_2, \qquad  \text{ for any }z\in\Nc(A), \text{ and any }v\in\Sigma_{\Fc,s},
\end{equation}
then for any $z_0\in\R^d$ and $y=A z_0+e$ where $\|e\|_2\leq\epsilon$, we have
\begin{equation}\label{equ:Fsnsp}
\|\hat z-z_0\|_2\leq\frac{2}{\nu_A}\left(\frac{\sqrt{N}}{c}+1\right)\epsilon+\frac{2}{c}\sigma_{\Fc,s}(z_0),
\end{equation}
where $\hat z$ is from \eqref{equ:l1s}. 
\end{corollary}

A strong null space property for recovery via the $\ell_1$ synthesis method was also proposed in \cite{CWW14}. We recall its definition here: A sensing matrix $A$
is said to have \emph{$F$ strong null space property of order $s$} if there exists a positive constant $c$ such that
\begin{equation}\label{equ:fsnsp2}
 \text{for  }\forall v\in\Nc(AF), |T|\leq s, \text{ }\exists u\in\Nc(F) \text{ satisfying } \|v_{T^c}\|-\|v_T+u\|_1\geq c\|Fv\|_2.
\end{equation}

It is not surprising that these two conditions \eqref{equ:fsnsp} and \eqref{equ:fsnsp2} are equivalent. See Lemma~\ref{lem:fsnsp} in the Appendix.

\begin{remark}\label{rem:snsp}
We compare Corollary \ref{cor:frame2} to \cite[Theorem 5.2]{CWW14} as they have the same assumption. In \cite[Theorem 5.2]{CWW14}, the conclusion is that $\|\hat z-z_0\|_2\leq\frac{2}{\nu_A}\left(\frac{\sqrt{N}}{c\nu_F}+1\right)\epsilon+\frac{2}{c}\min_{Fx=z_0,|T|\leq s}\|x_{T^c}\|_1$. First of all, the estimate in \eqref{equ:Fsnsp} avoids $\nu_F$ which is an improvement. Moreover,
the tail $\min_{Fx=z_0,|T|\leq s}\|x_{T^c}\|_1=\min_{Fx=z_0,|T|\leq s}\|x-x_{T}\|_1\geq\sigma_{F,s}(z_0)$ by \eqref{equ:tW}, so Corollary~\ref{cor:frame2} is stronger in the sense that it implies \cite[Theorem 5.2]{CWW14}, but not necessarily the other way around. Moreover, the proof of \cite[Theorem 5.2]{CWW14} is quite involved given the complicated definition \eqref{equ:fsnsp2}, whereas our analysis (proof of Theorem \ref{thm:snsp}) is simple and intuitive.
\end{remark}
Theorem \ref{thm:min2} reduces to the following corollary.

\begin{corollary}\label{cor:fmin}
  Let $\Fc$ be $s$-even. If for every $z_0\in\Sigma_{\Wc,s}$ and $y=\Ac z_0$, we have $\hat z=z_0$ where $\hat z$ is a solution of \eqref{equ:F} (with $\epsilon=0$), then the number of measurements $m$ satisfies
$$m\geq \frac{1}{8\ln3}s\ln\frac{N}{2s}.$$
\end{corollary}

We have mentioned earlier that $\Fc$ being $s$-even is not a harsh condition. For example, the matrix $F$ has $s$-NSP will guarantee that. Corollary \ref{cor:fmin} thus provides the minimum number of measurements needed using the $\ell_1$ synthesis method, which is first of its kind. It was proven in \cite{CCL18} that $O(s\ln(N/s))$ many subgaussian measurements do allow stable and robust recovery through \eqref{equ:F} if $F$ has $s$-NSP, so the bound in Corollary \ref{cor:fmin} is minimal and achievable.

When the atoms are the canonical orthonormal basis of $\R^d$, as shown in Example \ref{exa:cs},
we have $\|v\|_\Ec=\|v\|_1$ and $C_\Ec=\sqrt{d}$. 
In this case, Theorem \ref{thm:iff} reduces to a well known result and can be found in literature like \cite{CDD09} and \cite{FRbook}. Theorem \ref{thm:main} reduces to \cite[Theorem 2.4]{ACP11}.
Theorem \ref{thm:RNSP} is similar to \cite[Theorem 5]{F14}. Theorem \ref{thm:min2} and Theorem \ref{thm:min} in this setting can be found in \cite{FPRT10} or \cite{FRbook}.


\section{Low Rank Matrices Recovery}\label{sec:matrix}
This section addresses the case when  $\Mc=\{uv^T: u\in\R^{n_1}, v\in\R^{n_2}, \|u\|_2=\|v\|_2=1\}$ as mentioned in Example \ref{exa:matrix}. 
$\Sigma_{\Mc,s}=\{\sum_{i=1}^sc_iu_iv_i^T: c_i\in\R, \|u\|_2=\|v\|_2=1\}$ is the set of $n_1\times n_2$ matrices whose rank is at most~$s$. The constant $C_\Mc=\sqrt{n_1n_2}$. Let $K=\min\{n_1,n_2\}$. 

In this case, Definition \ref{def:wnsp} becomes 
\begin{equation}\label{equ:ranknsp}
\text{for every }Z\in\Nc(\Ac)\backslash\{0\} \text{ and every }V\in\Sigma_{\Mc,s}, \text{we have }\|V\|_*<\|Z-V\|_*,
\end{equation} 
which can be simplified to 
\begin{equation}\label{equ:ranknsp2}
\text{for every }Z\in\Nc(A)\backslash\{0\}, \text{we have }\sum_{i=1}^s\sigma_i(Z)<\sum_{i=s+1}^K\sigma_i(Z).
\end{equation}
The equivalence of \eqref{equ:ranknsp} and \eqref{equ:ranknsp2} is due to Lemma \ref{lem:matrix} in the Appendix. We  have the best $s$-term approximation of any matrix $Z$ to be
$$Z_s=\sum_{i=1}^s\sigma_i(Z)u_iv_i^T,$$
where $Z=\sum_{i=1}^K\sigma_i(Z)u_iv_i^T$ is its singular value decomposition.

Theorem \ref{thm:iff} reduces to the following corollary.
\begin{corollary}\label{cor:matrixiff}
The nuclear norm minimization problem
\begin{equation}\label{equ:nucmin}
\hat{Z}=\arg\min_{Z\in \R^{n_1\times n_2}} \|Z\|_{*} \quad\text{ subject to } \quad \Ac Z=y.
\end{equation}
is successful at recovering all signals in $\Sigma_{\Mc,s}$ if and only if \eqref{equ:ranknsp2} holds.
\end{corollary}
   
Corollary \ref{cor:matrixiff} can be found in \cite[Theorem 1.1]{RXH08} or \cite[Theorem 15]{KKRT16} as well, but our general tail definition provides a more concise proof than the one in \cite{RXH08}.



The recovery scheme \eqref{equ:K} becomes
\begin{equation}\label{equ:nucmine}
\hat{Z}=\arg\min_{Z\in \R^{n_1\times n_2}} \|Z\|_{*} \quad\text{ subject to } \quad \|\Ac Z-Y\|_F\leq\epsilon.
\end{equation}

We also rewrite Theorem \ref{thm:main}, Theorem \ref{thm:RNSP}, and Theorem \ref{thm:snsp} in this case as one corollary, as $\Mc$ is splittable with constant 1.

\begin{corollary}\label{cor:matrix}

(a) If $\Ac$ has the stable $\Mc$ null space property of order $s$ as 
\begin{equation}
\|Z_s\|_*\leq\rho\|Z-Z_s\|_*, \qquad  \text{ for any }z\in\Nc(A),
\end{equation}
then for any $Z_0\in\R^{n_1\times n_2}$ and $Y=\Ac Z_0+e$ where $\|e\|_F\leq\epsilon$, we have
$$\|\hat Z-Z_0\|_*\leq\frac{2+2\rho}{1-\rho}\|Z-Z_s\|_*+\frac{4\sqrt{n_1n_2}}{(1-\rho)\nu_\Ac}\epsilon,$$
where $\hat Z$ is from \eqref{equ:nucmine}.

(b) If $A$ has the robust $\Mc$ null space property of order $s$ as
\begin{equation}\label{equ:matrixrnsp}
\|Z_s\|_\Fc\leq\rho\|Z-Z_s\|_*+\tau \|\Ac Z\|_2, \qquad \text{for any }Z\in\R^{n_1\times n_2}, 
\end{equation}
then for any $Z_0\in\R^{n_1\times n_2}$ and $Y=\Ac Z_0+e$ where $\|e\|_F\leq\epsilon$, we have
\begin{equation}\label{equ:matrixb}
\|\hat Z-Z_0\|_*\leq\frac{2(1+\rho)}{1-\rho}\|Z-Z_s\|_*+\frac{4\tau}{1-\rho}\epsilon,
\end{equation}
where $\hat Z$ is from \eqref{equ:nucmine}. 

(c) If $A$ has the strong $\Mc$ null space property of order $s$ as
\begin{equation}\label{equ:msnsp}
\|Z-Z_s\|_*-\|Z_s\|_*\geq c\|Z\|_F, \qquad  \text{ for any }Z\in\Nc(A),
\end{equation}
then for any $Z_0\in\R^{n_1\times n_2}$ and $Y=\Ac Z_0+e$ where $\|e\|_F\leq\epsilon$, we have
\begin{equation}
\|\hat Z-Z_0\|_2\leq\frac{2}{\nu_A}\left(\frac{\sqrt{n_1n_2}}{c}+1\right)\epsilon+\frac{2}{c}\|Z-Z_s\|_*,
\end{equation}
where $\hat Z$ is from \eqref{equ:nucmine}. 
\end{corollary}

Corollary \ref{cor:matrix}(a)(c) are  new results. Note that the strong null space property is simplifed to \eqref{equ:msnsp} due to Lemma \ref{lem:matrix}. Corollary \ref{cor:matrix}(b) is essentially the same as \cite[Theorem 11]{KKRT16}. Compared to \cite[Theorem 11]{KKRT16}, the lack of $\sqrt{s}$ in \eqref{equ:matrixb} is due to the lack of the same term in our robust null space property \eqref{equ:matrixrnsp}.


The special case of Example \ref{exa:p} has a lot of similarities with low rank matrix recovery. We leave it to the readers to modify Corollary \ref{cor:matrix} for recovery results using \eqref{equ:P}.


\section{Appendix}
\begin{lemma}\label{lem:nuc}
 Let $\Mc=\{uv^T: u\in\R^{n_1}, v\in\R^{n_2}, \|u\|_2=\|v\|_2=1\}$, then $\|X\|_\Mc=\|X\|_*$.
\end{lemma}
\begin{proof}
Let $X=\sum_{i=1}^K \sigma_iu_iv_i^T$ be its singular value decomposition.
Suppose we can rewrite it as $X=\sum s_ip_iq_i^T$ where $\|p_i\|_2=\|q_i\|_2=1, s_i>0$.

$\sum_{i=1}^K\sigma_i=\|X\|_*=\|\sum_{i=1}^m s_iu_iv_i^T\|_*\leq\sum_{i=1}^m\|s_iu_iv_i^T\|_*=\sum_{i=1}^ms_i$.

By the definition of atomic norm, we directly have $\|X\|_\Mc=\sum_{i=1}^m\sigma_i$.
\end{proof}

The following lemma is a consequence of \cite[Theorem 7.4.51]{HJ}.
\begin{lemma}\label{lem:sv}
$\sum_{i=1}^K|\sigma_i(X)-\sigma_i(Y)|\leq \sum_{i=1}^K\sigma_i(X-Y)$
\end{lemma}


\begin{lemma}\label{lem:split}
Both $\Ec$ and $\Mc$ are splittable with $\beta=1$.
\end{lemma}
\begin{proof}
For $\Ec$, given $x,y\in\R^d$, let $T$ be the index set such that $x_s=x_T$. It is clear that $-\|y_T\|_1\geq-\|y_s\|_1$ and $\|y_{T^c}\|_1\geq\|y-y_s\|_1$. Then
\begin{align*}
\|x+y\|_1&=\|x_T+y_T\|_1+\|x_{T^c}+y_{T^c}\|_1\\
&\geq\|x_T\|_1-\|y_T\|_1+\|y_{T^c}\|_1-\|x_{T^c}\|_1\\
&\geq\|x_s\|_1-\|y_s\|_1+\|y-y_s\|_1-\|x-x_s\|_1.
\end{align*}
For $\Mc$, given $X, Y\in\R^{n_1\times n_2}$
\begin{align*}
\|X+Y\|_*&=\sum_{i=1}^K\sigma_i(X-(-Y))\geq \sum_{i=1}^K|\sigma_i(X)-\sigma_i(Y)|\\
&=\sum_{i=1}^s|\sigma_i(X)-\sigma_i(Y)|+\sum_{i=s+1}^K|\sigma_i(X)-\sigma_i(Y)|\\
&\geq \sum_{i=1}^s\sigma_i(X)-\sum_{i=1}^s\sigma_i(Y)+\sum_{i=s+1}^K\sigma_i(Y)-\sum_{i=s+1}^K\sigma_i(X)\\
&=\|X_s\|_*-\|Y_s\|_*+\|Y-Y_s\|_*-\|X-X_s\|_*
\end{align*}
\end{proof}


\begin{lemma}\label{lem:snsp}
Let $F=[f_1,f_2,\cdots, f_N]$ has null space property of order $s$ (see \eqref{equ:nsp}), then the atomic set $\Fc=\{f_1,f_2,\cdots,f_N\}$ is $s$-even.

\end{lemma}
\begin{proof}
Let $z=\sum_{j\in J}c_jf_j=\sum_{i\in[N]}\alpha_if_i$ be two different representations of $z$ in $\Fc$. In the first representation, we have $|c_j|=1$ and $|J|\leq s$. We have $\sum_{j\in J}(\alpha_j-c_j)f_j+\sum_{j\notin J}\alpha_jf_j=0$. By the null space property of $F$, 
$\sum_{j\notin J}|\alpha_j|>\sum_{j\in J}|\alpha_j-c_j|\geq\sum_{j\in J}|c_j|-\sum_{j\in J}|\alpha_j|,$
so $\sum_{j\in [N]}|\alpha_j|>\sum_{j\in J}|c_j|$. By definition of atomic norm, we have $\|z\|_\Fc=\sum_{j\in J}|c_j|$.
\end{proof}

\begin{proof}[Proof of Theorem \ref{thm:min}]
Let $\Uc$ be a family of subsets of $[N]$ given by Lemma \ref{lem:N} with $t=\lfloor s/2\rfloor$. For each $I\in\Uc$, define $z_I:=\frac{1}{t}\sum_{i\in I}w_i$. We have $|I|=\lfloor s/2\rfloor$. By definition, $\|z_I\|_\Wc\leq1$.

For $I\neq J$,  $z_I-z_J=\frac{1}{t}(\sum_{i\in I-J}w_i-\sum_{j\in J-I}w_j)$. Its support $|I-J|+|J-I|=|I|+|J|-2|I\cap J|$, so $s\geq2t\geq|I-J|+|J-I|>t$. Since $\Wc$ is $s$-even, $\|z_i-z_j\|_\Wc=\frac{1}{t}(|I-J|+|J-I|)>1$.

We claim that $\{\Ac(z_I+\rho B_\Wc), I\in\Uc\}$ is a disjoint collection of subsets of $\Ac(\R^d)$, where $\rho=\frac{1}{2C+1}$. Otherwise there exist $I\neq J$ and $v, v'\in B_\Wc$ such that $\Ac(z_I+\rho v)=\Ac(z_J+\rho v')$, resulting
\begin{align*}
1&<\|z_I-z_J\|_\Wc=\|z_I+\rho v-\Delta(\Ac(z_I+\rho v))-z_J+\Delta(\Ac(z_J+\rho v'))-\rho v+\rho v'-\rho v'\|_\Wc\\
&\leq\|z_I+\rho v-\Delta(\Ac(z_I+\rho v))\|_\Wc+\|\rho v\|_\Wc+\|\rho v'\|_\Wc+\|\Delta(\Ac(z_J+\rho v'))-(z_J+\rho v')\|_\Wc\\
&\leq C\sigma_{\Wc,s}(z_I+\rho v)+2\rho+C\sigma_{\Wc,s}(z_J+\rho v')\leq C\|\rho v\|_\Wc+2\rho+C\|\rho v'\|_\Wc\leq1.
\end{align*}

It is easy to see that $\Ac(z_I+\rho B_\Wc)\subset (1+\rho)\Ac B_\Wc$ for any $I\in\Uc$. So
\begin{align}\label{equ:vol}
|\Uc|\vol(\Ac(\rho(B_\Wc)))\leq \vol((1+\rho)\Ac B_\Wc)
\end{align}
Let $\Ac B_\Wc$ have  dimension $r\leq m$, so \eqref{equ:vol} becomes
\begin{align*}
&\left(\frac{N}{4t}\right)^{t/2}\rho^r\vol(\Ac B_\Wc)\leq(1+\rho)^r\vol(\Ac B_\Wc)\\
\Rightarrow & \left(\frac{N}{4t}\right)^{t/2}\leq(1+\frac{1}{\rho})^r\leq(2C+3)^m.
\end{align*}
Taking log of both sides yields
$$m\ln(2C+3)\geq\frac{t}{2}\ln\frac{N}{4t}\geq\frac{s/2-1}{2}\ln\frac{N}{2s}\geq\frac{s}{8}\ln\frac{N}{2s},$$ if $s>2$.
\end{proof}

\begin{lemma}\label{lem:fsplit}
The frame $\{f_n=(\cos\frac{2\pi (n-1)}{8},\sin\frac{2\pi (n-1)}{8})\}_{n=1}^8$ of $\R^2$ is 1-splittable with constant $\sqrt{2}-1$.

\end{lemma}

\begin{proof}
We need to prove that \begin{equation}\label{equ:75}
\|x+y\|_\Fc\geq\|x_1\|_\Fc-\|y_1\|_\Fc+(\sqrt{2}-1)(\|y-y_1\|_\Fc-\|x-x_1\|_\Fc).
\end{equation}
Recall $x_1=\argmin\limits_{v\in\Sigma_{\Fc,1}}\|x-v\|_{\Fc}$.

First, we notice that given any $a\in[0,1]$, we have $\|af_n+(1-a)f_{n+1}\|_\Fc=1$ for any $n$ since  $af_n+(1-a)f_{n+1}$ is on the boundary of the convex hull. Therefore with scaling, we have 
\begin{equation}\label{equ:f1f2}
\|af_n+bf_{n+1}\|_\Fc = a+b, \quad\text{for any }a\geq0, b\geq0.
\end{equation}

Next, we  establish that 
\begin{equation}\label{equ:f1f3}
\|af_1+bf_3\|_\Fc=\max\{|a|,|b|\}+\beta\min\{|a|,|b|\}\geq |a|+\beta |b|,
\end{equation}
where $\beta=\sqrt{2}-1$. The inequality is a straightforward calculation. For the equality, due to symmetry, we only need to prove the case when $a>b>0$. In this case, $x=af_1+bf_3=af_1+b (\frac{2}{\sqrt{2}}f_2-f_1)=(a-b)f_1+\frac{2b}{\sqrt{2}}f_2$. By \eqref{equ:f1f2}, we have $\|x\|_\Fc=a-b+\sqrt{2}b=a+\beta b$. Note that \eqref{equ:f1f3} is also true if the vector is $af_2+bf_4$ since these two vectors are also orthogonal.

Third, we will show for $x=af_1+bf_3=\frac{a+b}{\sqrt{2}}f_2+\frac{b-a}{\sqrt{2}}f_4$ where $a>b>0$, 
\begin{equation}\label{equ:cases}
\begin{array}{lll}&x_1 = af_1, & \text{if } a>b/\beta,\\
&x_1=\frac{a+b}{\sqrt{2}}f_2, & \text{if } a\leq b/\beta.
\end{array}\end{equation}
The idea is that to find the best 1-sparse approximation of $x$, we simply find which frame vector is closest to $x$ (biggest inner product). In the case of $a>\beta b$, $x$ is closest to $f_1$, so the best 1-sparse approximation of $x$ is $af_1$. To prove this case rigorously, we can compute that other 1-sparse approximations are providing bigger tails. If we try to take out a portion of $f_2$, then $\min_c\|x-cf_2\|_\Fc\stackrel{\eqref{equ:f1f3}}{=}\|x-\frac{a+b}{\sqrt{2}}f_2\|_\Fc=\frac{a-b}{\sqrt{2}}>b$. Similarly, $\min_c\|x-cf_2\|_\Fc=a>b$, and $\min_c\|x-cf_4\|_\Fc=\frac{a+b}{\sqrt{2}}>b$. The $a\leq{b/\beta}$ case in \eqref{equ:cases} can be shown similarly.

A byproduct of the above argument is that for $x=af_1+bf_3$,
\begin{equation}\label{equ:x1}
\|x_1\|_\Fc\geq\max\{ |a|,|b|\}.
\end{equation}

Finally, to prove \eqref{equ:75}, we assume, without loss of generality, that $x=af_1+bf_3$ with $a>{b/\beta}>0$, and $y=cf_1+df_3$ with arbitrary $c,d$. Therefore $\|x_1\|_\Fc=a$, and $\|x-x_1\|_\Fc=b$. So
\begin{align*}
\|x+y\|_\Fc&=\|(a+c)f_1+(b+d)f_3\|_\Fc\stackrel{\eqref{equ:f1f3}}{\geq}|a+c|+\beta|b+d|\\
&\geq a-|c|+\beta(|d|-b)\\
&\stackrel{\eqref{equ:x1}}{\geq} \|x_1\|_\Fc-\|y_1\|_\Fc+\beta(\|y-y_1\|_\Fc-\|x-x_1\|_\Fc)
\end{align*}
as desired.
\end{proof}

Note that the equality in \eqref{equ:75} can be achieved, for example, at $x=6f_1+f_3, y=-5f_1-2f_3$.

\color{black}
\begin{lemma}\label{lem:fsnsp}
The two conditions \eqref{equ:fsnsp} and \eqref{equ:fsnsp2} are equivalent.
\end{lemma}
\begin{proof}
(\eqref{equ:fsnsp2}$\Longrightarrow$\eqref{equ:fsnsp}) For any $z\in\Nc(A)$ and any $v\in\Sigma_{F,s}$, by \eqref{equ:z-v}, let $\|z-v\|_\Fc=\|x-u_1\|_1$ where $z=Fx, v=Fu_1=Fu_2$ where $\|u_2\|_0\leq s$. Let $T$ be the support of $u_2$. The vector $y=x-u_1+u_2\in\Nc(AF)$. By assumption, there exists $u\in\Nc(F)$ such that $\|y_{T^c}\|_1-\|y_T+u\|_1\geq c\|Fy\|_2=c\|z\|_2$. Now we have
\begin{align*}
\|z-v\|_\Fc-\|v\|_\Fc&\geq\|x-u_1\|_1-\|u_2+u\|_1=\|y-u_2\|_1-\|u_2+u\|_1\\
&=\|y_T-u_2\|_1+\|y_{T^c}\|_1-\|u_2+u\|_1\\
&\geq\|y_T-u_2\|_1+\|y_T+u\|_1+c\|z\|_2-\|u_2+u\|_1\geq c\|z\|_2.
\end{align*}

(\eqref{equ:fsnsp}$\Longrightarrow$\eqref{equ:fsnsp2})
 For any $x\in\Nc(AF)$ and any $|T|\leq s$,  let $v=Fx_T\in\Sigma_{F,s}$ and $\|v\|_\Fc=\|y\|_1$ where $v=Fy$. Since $z=Fx\in\Nc(A)$, by \eqref{equ:fsnsp}, $\|z-v\|_\Fc-\|v\|_\Fc\geq c\|z\|_2$.
 
 Clearly $y-x_T\in\Nc(F)$, and
$\|x_{T^c}\|_1-\|x_T+y-x_T\|_1\geq\|z-v\|_\Fc-\|v\|_\Fc\geq c\|z\|_2=c\|Fx\|_2$.
\end{proof}



\begin{lemma}\label{lem:matrix}
For any $Z\in\R^{n_1\times n_2}$, 
$$\min_{V\in\Sigma_{\Mc,s}}\left(\|Z-V\|_*-\|V\|_*\right)=\|Z-Z_s\|_*-\|Z_s\|_*=-\sum_{i=1}^s\sigma_i(Z)+\sum_{i=s+1}^K\sigma_i(Z).$$
\end{lemma}
\begin{proof}
Let $\sigma_i$ be the singular values of $Z$.

$\|Z-V\|_*-\|V\|_*\stackrel{Lemma\ \ref{lem:sv}}{\geq}\sum_{i=1}^K|\sigma_i-\sigma_i(V)|-\sum_{i=1}^s\sigma_i(V)=\sum_{i=1}^s|\sigma_i-\sigma_i(V)|+\sum_{i=s+1}^K\sigma_i-\sum_{i=1}^s\sigma_i(V)$.

It is easy to show that $\min_{c\in\R}(|\sigma-c|-c)=-\sigma$, therefore we have
$$\|Z-V\|_*-\|V\|_*\geq\sum_{i=1}^s(|\sigma_i-\sigma_i(V)|-\sigma_i(V))+\sum_{i=s+1}^K\sigma_i\geq-\sum_{i=1}^s\sigma_i+\sum_{i=s+1}^K\sigma_i.$$
\end{proof}

\section*{Acknowledgments}
The author is funded by NSF DMS-2050028. 

%
%
%
%

\end{document}